\documentclass[11pt,letterpaper]{article}
\usepackage{times}
\usepackage{graphicx}
\usepackage{fullpage}
\usepackage[numbers]{natbib}

\usepackage{url}            
\usepackage{dsfont}       
\usepackage{graphicx}
\usepackage{bm}
\usepackage{amsmath,amsthm,amssymb,amsfonts}
\newtheorem{theorem}{Theorem}
\newtheorem{proposition}{Proposition}
\newtheorem{lemma}{Lemma}
\newtheorem{corollary}{Corollary}
\theoremstyle{definition}

\usepackage[ruled,longend]{algorithm2e}
\usepackage{subcaption}
\usepackage{enumitem}
\usepackage[numbers]{natbib}

\newcommand{\R}{\mathds{R}}

\DeclareMathOperator*{\argmin}{arg\,min}

\newcommand{\agents}{N}
\newcommand{\honest}{H}
\newcommand{\manip}{M}
\renewcommand{\tilde}{\widetilde}
\renewcommand{\bar}{\overline}
\renewcommand{\hat}{\widehat}
\renewcommand{\vec}{\bm}
\newcommand{\vx}{\vec{X}}
\newcommand{\vH}{\vec{H}}
\newcommand{\vxi}{\vec{x_i}}
\newcommand{\vxj}{\vec{x_j}}
\newcommand{\vxk}{\vec{x_k}}
\newcommand{\vy}{\vec{y}}
\newcommand{\vty}{\vec{\tilde{y}}}
\newcommand{\vby}{\vec{\bar{y}}}
\newcommand{\vhy}{\vec{\hat{y}}}
\newcommand{\ty}{\tilde{y}}
\newcommand{\tyi}{\tilde{y}_i}
\newcommand{\tyj}{\tilde{y}_j}
\newcommand{\tyk}{\tilde{y}_k}
\newcommand{\byi}{\bar{y}_i}
\newcommand{\by}{\bar{y}}
\newcommand{\byh}{\bar{y}_h}
\newcommand{\byne}{\bar{y}_{ne}}

\newcommand{\vbeta}{\vec{\beta}}
\newcommand{\vbetast}{\vec{\beta^*}}
\newcommand{\set}[1]{\{#1\}}
\newcommand{\med}{\mathrm{med}}
\newcommand{\NE}{\mathrm{NE}}
\newcommand{\PPoA}{\mathrm{PPoA}}
\newcommand{\br}{\mathrm{br}}
\renewcommand{\paragraph}[1]{\medskip\noindent\emph{#1}}
\newcount\Comments  
\usepackage{color} 
\newcommand{\kibitz}[2]{\ifnum\Comments=1\textcolor{#1}{#2}\fi}

\begin{document}

\title{The Effect of Strategic Noise in Linear Regression\thanks{During the course of this work, Shah was partially supported by an NSERC Discovery Grant.}}
\author{Safwan Hossain\\University of Toronto\\\texttt{safwan.hossain@mail.utoronto.ca} \and Nisarg Shah\\University of Toronto\\\texttt{nisarg@cs.toronto.edu}}
\date{}

\maketitle

\begin{abstract}
We build on an emerging line of work which studies strategic manipulations in training data provided to machine learning algorithms. Specifically, we focus on the ubiquitous task of linear regression. Prior work focused on the design of strategyproof algorithms, which aim to prevent such manipulations altogether by aligning the incentives of data sources. However, algorithms used in practice are often not strategyproof, which induces a strategic game among the agents. We focus on a broad class of non-strategyproof algorithms for linear regression, namely $\ell_p$ norm minimization ($p > 1$) with convex regularization. We show that when manipulations are bounded, every algorithm in this class admits a unique pure Nash equilibrium outcome. We also shed light on the structure of this equilibrium by uncovering a surprising connection between strategyproof algorithms and pure Nash equilibria of non-strategyproof algorithms in a broader setting, which may be of independent interest. Finally, we analyze the quality of equilibria under these algorithms in terms of the price of anarchy.
\end{abstract}

\section{Introduction}

Linear regression aims to find a linear relationship between explanatory variables and response variables. Under certain assumptions, it is known that minimizing a suitable loss function on training data generalizes well to unseen test data \citep{statistical_Learning}. However, traditional analysis assumes that the algorithm has access to untainted data drawn from the underlying distribution. Relaxing this assumption, a significant body of recent work has focused on making machine learning algorithms robust to stochastic or adversarial noise; the former is too benign \citep{Litt88,GS95,frenay2013classification,natarajan2013learning}, while the latter is too pessimistic \citep{KL93,BEK02,CCM13,gu2014towards}. A third model, more recent and prescient, is that of \textit{strategic noise}, which is a game-theoretic modeling of noise that sits in between the two. Here, it is assumed that the training set is provided by self-interested agents, who may manipulate to minimize loss on their own data.

We focus on strategic noise in linear regression. \citet{dekel2010incentive} provide an example of retailer Zara, which uses regression to predict product demand at each store, partially based on self-reported data provided by the stores. Given limited supply of popular items, store managers may engage in strategic manipulation to ensure the distribution process benefits them, and there is substantial evidence that this is widespread \citep{caro2010zara}. Strategic behavior by even a small number of agents can significantly affect the overall system, including agents who have not participated in such behavior. Prior work has focused on designing \emph{strategyproof} algorithms for linear regression \citep{perote2004strategy,dekel2010incentive,chen2018strategyproof}, under which agents provably cannot benefit by misreporting their data. While strategyproofness is a strong guarantee, it is only satisfied by severely restricted algorithms. Indeed, as we observe later in the paper, most practical algorithms for linear regression are \emph{not} strategyproof. 

When strategic agents with competing interests manipulate the input data under a non-strategyproof algorithm, a game is induced between them. Game theory literature offers several tools to analyze such behaviour, such as Nash equilibria and the price of anarchy~\cite{nisan2007algorithmic}. We use these tools to answer three key questions:

\begin{itemize}
	\setlength\itemsep{0.3em}
	\item Does the induced game always admit a pure Nash equilibrium?
	\item What are the characteristics of these equilibria?
	\item Is there a connection between strategyproof algorithms and equilibria of non-strategyproof algorithms?
\end{itemize}

We consider linear regression algorithms which minimize the $\ell_p$-norm of residuals (where $p > 1$) with convex regularization. This class includes most popular linear regression algorithms, including the ordinary least squares (OLS), lasso, group lasso, ridge regression, and elastic net regression. Our key result is that the game induced by an algorithm in this class has three properties: a) it always has a pure Nash equilibrium, b) all pure Nash equilibria result in the same regression hyperplane, and c) there exists a strategyproof algorithm which returns this equilibrium regression hyperplane given non-manipulated data. We also analyze the quality of this equilibrium outcome, measured by the pure price of anarchy. We show that for a broad subset of algorithms in this class, the pure price of anarchy is unbounded. 

\subsection{Related Work}
A special case of linear regression is facility location in one dimension~\cite{moulin1980strategy}, where each agent $i$ is located at some $y_i$ on the real line. An algorithm elicits the preferred locations of the agents (who can misreport) and chooses a location $\bar{y}$ to place a facility. A significant body of literature in game theory is devoted to understanding strategyproof algorithms in this domain~\cite{moulin1980strategy,caragiannis2016truthful}, which includes placing the facility at the median of the reported locations. A more recent line of work studies equilibria of non-strategyproof algorithms such as placing the facility at the average of the reported locations~\cite{renault2005protecting,renault2011assessing,yamamura2013generalized}. Similarly, in the more general linear regression setting, prior work has focused on strategyproof algorithms~\cite{perote2004strategy,dekel2010incentive,chen2018strategyproof}. We complete the picture by studying equilibria of non-strategyproof algorithms for linear regression. 

We use a standard model of strategic manipulations in linear regression~\cite{perote2004strategy,dekel2010incentive,chen2018strategyproof}. \citet{perote2004strategy} designed a strategyproof algorithm in two dimensions. \citet{dekel2010incentive} proved that least absolute deviations (LAD), which minimizes the $\ell_1$-norm of residuals without regularization, is strategyproof. \citet{chen2018strategyproof} extended their result to include regularization, and designed a new family of strategyproof algorithms in high dimensions. They also analyzed the loss in mean squared error (MSE) under a strategyproof algorithm as compared to the OLS, which minimizes MSE. They showed that any strategyproof algorithm has at least twice as much MSE as the OLS in the worst case, and that this ratio is $\Theta(n)$ for LAD. Our result (Theorem~\ref{theorem:regression_sc_unbounded}) shows that the ratio of the equilibrium MSE under the algorithms we study to the optimal MSE of the OLS is unbounded. Through the connection we establish to strategyproof algorithms (Theorem~\ref{theorem:NE_is_strategyproof}), this also implies unbounded ratio for the broad class of corresponding strategyproof algorithms. 

Finally, we mention that strategic manipulations have been studied in various other machine learning contexts, e.g., manipulations of feature vectors \citep{HMPW16,DRZWW18}, strategic classification~\cite{MPR12,HMPW16,DRZWW18}, competition among different algorithms~\cite{mansour2017competing,immorlica2011dueling,ben2017best,ben2019regression}, or manipulations due to privacy concerns~\cite{CIL15,CDP15}.  

\section{Model}\label{section:Model}
In linear regression, we are given $n$ training data points of the form $(\vxi,y_i)$, where $\vxi \in \R^d$ are the explanatory variables, and $y_i \in \R$ is the response variable.\footnote{In the regression literature, these are also called independent and dependent variables, respectively. Following the standard convention, we assume that the last component of each $\vxi$ is a constant, say $1$.} Let $\vx \in \R^{n \times d}$ be the matrix with $\vxi$ as its $i^{\text{th}}$ column, and $\vy = (y_1,\ldots,y_n)$. The goal of a linear regression algorithm is to find a hyperplane with normal vector $\vbeta$ such that $\vbeta^T \vxi$ is a good estimate of $y_i$. The residual of point $i$ is $r_i = |y_i - \vbeta^T x_i|$. 

\paragraph{\textbf{Algorithms:}} We focus on a broad class of algorithms parametrized by $p > 1$ and a regularizing function $R:\R^d \to \R$. The $(p,R)$-regression algorithm minimizes the following loss function over $\vbeta$:
\begin{equation}\label{equation:loss_fcn}
\smash{\mathcal{L}(\vy,\vx,\vbeta) = {\textstyle\sum_{i=1}^{n}}{|y_i - \vbeta^T \vxi|^p} + R(\vbeta).}
\end{equation}
We assume that $R$ is convex and differentiable. For $p > 1$, this objective is strictly convex, admitting a unique optimum $\vbeta^*$. When there is no regularization, we refer to it as the $(p,0)$-regression algorithm.

\paragraph{\textbf{Strategic model:}} We follow a standard model of strategic interactions studied in the literature~\cite{perote2004strategy,dekel2010incentive,chen2018strategyproof}. A training data point $(\vxi,y_i)$ is provided by an agent $i$. $\agents = [n] := \set{1,\ldots,n}$ denotes the set of all agents. $\vxi$ is public information, which is non-manipulable, but $y_i$ is held private by agent $i$. We assume a subset of agents $\honest \subset \agents$ (with $h=|\honest|$) are honest and always report $\tyi = y_i$. The remaining agents in $\manip = \agents \setminus \honest$ (with $m=|\manip|$) are strategic and may report $\tyi \neq y_i$. Note that we allow all agents in $N$ be strategic; that is, we allow $H = \emptyset$ and $M = N$. For convenience, we assume that $\manip = [m]$ and $\honest = \set{m+1,\ldots,n}$. However, we emphasize that our algorithms do not know which agents are strategic and which are honest. Given a set of reports $\vty$, honest agents' reports are denoted by $\vty_{\honest}$ (note that $\vty_{\honest}=\vy_{\honest}$) and strategic agents' reports by $\vty_{\manip}$. In accordance with related literature, we focus our analysis to the training set and do not consider strategic manipulation in test data, leaving this for future work.  

The $(p,R)$-regression algorithm takes as input $\vx$ and $\vty$, and returns $\vbetast$ minimizing the loss in Equation~\eqref{equation:loss_fcn}. We say that $\byi = (\vbetast)^T\vxi$ is the \emph{outcome} for agent $i$. Since $\vx$ and $\vy_{\honest}$ are non manipulable, we can treat them as fixed. Hence, $\vty_{\manip}$ is the only input which matters, and $\vby_{\manip}$ is the output for these manipulating agents. For an algorithm $f$, we use the notation $f(\vty_{\manip}) = \vby_{\manip}$, and let $f_i$ denote the function returning agent $i$'s outcome $\byi$. A strategic agent $i$ manipulates to ensure this outcome is as close to her true response variable $y_i$ as possible. Formally, agent $i$ has \emph{single-peaked preferences} $\succeq_i$ (with strict preference denoted by $\succ_i$) over $\byi$ with peak at $y_i$. That is, for all $a < b \le y_i$ or $a > b \ge y_i$, we have $b \succ_i a$. Agent $i$ is \emph{perfectly happy} when $\byi = y_i$. In this work, we assume that for each agent $i$, both $y_i$ and $\tyi$ are bounded (WLOG, say they belong to $[0,1]$).

\paragraph{\textbf{Nash equilibria:}} This strategic interaction induces a game among agents in $\manip$, and we are interested in the pure Nash equilibria (PNE) of this game. We say that $\vty_{\manip}$ is a \emph{Nash equilibrium} (NE) if no strategic agent $i \in \manip$ can strictly gain by changing her report, i.e., if $\forall i, \, \forall \, \tyi'$, $f_i(\vty_{\manip}) \succeq_i f_i(\tyi',\vty_{\manip\setminus\{i\}})$. We say that $\vty_{\manip}$ is a \emph{pure Nash equilibrium} (PNE) if it is a NE and each $\tyi$ is deterministic. Let $\NE_f(\vy)$ denote the set of pure Nash equilibria under $f$ when the peaks of agents' preferences\footnote{ Equilibria can generally depend on the full preferences, but results in Section~\ref{section:regression} show only peaks matter.} are given by $\vy$. For $\vhy_{\manip} \in \NE_f(\vy)$, let $f(\vhy_{\manip})$ be the corresponding \emph{PNE outcome}.

\paragraph{\textbf{Strategyproofness:}} We say that an algorithm $f$ is \emph{strategyproof} if no agent can benefit by misreporting her true response variable regardless of the reports of the other agents, i.e., $\forall i,\, \forall \vty_{\manip}$, $f_i(y_i,\vty_{\manip\setminus\{i\}})$ $\succeq_i$ $f_i(\vty_{\manip})$. Note that strategyproofness implies that each agent reporting her true value (i.e. $\vty_{\manip} = \vy_{\manip}$) is a pure Nash equilibrium.

\paragraph{\textbf{Pure price of anarchy (PPoA):}} It is natural to measure the cost of selfish behavior on the overall system. A classic notion is the \emph{pure price of anarchy} (PPoA)~\cite{koutsoupias1999worst,nisan2007algorithmic}, which is defined as the ratio between the maximum social cost under any PNE and the optimal social cost under honest reporting, for an appropriate measure of social cost. Here, social cost is a measure of the overall fit. In regression, it is typical to measure fit using the $\ell_q$ norm of absolute residuals for some $q$. While we study the equilibrium of $\ell_p$ regression mechanisms for different $p$ values, we need to evaluate them using a single value of $q$, so that the results are comparable. For our theoretical analysis, we use mean squared error (which corresponds to $q=2$) since it is the standard measure of fit in the literature~\cite{chen2018strategyproof}. One way to interpret our results is: \emph{If our goal were to minimize the MSE, which $\ell_p$ regression mechanism would we choose, assuming that the strategic agents would achieve equilibrium?} We also present empirical results for other values of $q$. Slightly abusing the notation by letting $f$ map all reports to all outcomes (not just for agents in $\manip$), we write: 
$$
\PPoA(f) = \max_{\vy \in [0,1]^n} \frac{\max_{\vhy \in \NE_f(\vy)} {\textstyle\sum_{i=1}^n} |y_i-f_i(\vhy)|^2}{{\textstyle\sum_{i=1}^n} |y_i-\byi^\text{OLS}|^2},
$$
where $\vby^{\text{OLS}}$ is the outcome of OLS (i.e. the $(2,0$)-regression algorithm) under honest reporting, which minimizes mean squared error. Note that the PPoA, as we have defined it, measures the impact of the behavior of strategic agents on all agents, including on the honest agents.        

\section{Warm-Up: The 1D Case}\label{section:1d}
As a warm-up, we review the more restricted facility location setting in one dimension. Here, each agent $i$ has an associated scalar value $y_i \in [0,1]$ and the algorithm must produce the same outcome for all agents (i.e. $\bar{y}_i = \bar{y}_j \, \forall \, i,j \in \agents$). Hence, the algorithm is a function $f:[0,1]^{m} \to \R$. This is a special case of linear regression where agents have identical explanatory variables.

Much of the literature on facility location has focused on strategyproof algorithms. \citet{moulin1980strategy} showed that an algorithm $f$ is strategyproof and anonymous\footnote{This is a mild condition which requires treating the agents symmetrically.} if and only if it is a \emph{generalized median} given by $f(y_1,\ldots,y_{n}) = \med(y_1,\ldots,y_{n},$ $\alpha_0,\ldots,\alpha_n)$, where $\med$ denotes the median and $\alpha_k$ is a fixed constant (called a \emph{phantom}) for each $k$. \citet{caragiannis2016truthful} focused on a notion of worst-case statistical efficiency, and provided a characterization of generalized medians which exhibit optimal efficiency. In particular, they showed that the \emph{uniform generalized median} given by $f(y_1,\ldots,y_{n}) = \med(y_1,\ldots,y_{n},0,1/n,2/n,\ldots,1)$ is has optimal statistical efficiency. 

A more recent line of literature has focused on manipulations under non-strategyproof rules. Recall that under a non-strategyproof rule $f$, each strategic agent $i \in \manip$ reports a value $\tyi$, which may be different from $y_i$. For the facility location setting, the $(p,R)$-regression algorithm described in Section~\ref{section:Model} reduces to $f(\tilde{y}_1,\ldots,\tilde{y}_n) = \argmin_y \sum_{i=1}^m |\tilde{y}_i-y|^p + \sum_{i=m+1}^n |y_i-y|^p + R(y)$. For $p=1$, this is known to be strategyproof~\cite{chen2018strategyproof}. When $p > 1$, which is the focus of our work, this rule is not strategyproof, as we observe in Section~\ref{section:regression}. 

In this family, the most natural rule is the \emph{average rule} given by $f(\tilde{y}_1,\ldots,\tilde{y}_n)$ $= (1/n) \sum_{i=1}^n \tyi$. This corresponds to $p=2$ with no honest agents or regularization. For this rule, \citet{renault2005protecting} showed that there is always a pure Nash equilibrium, and the pure Nash equilibrium outcome is unique. This outcome is given by $\med(y_1,\ldots,y_{n},0,1/n,\ldots,1)$, which coincides with the outcome of the uniform generalized median, which is strategyproof.

Generalizing this result, \citet{yamamura2013generalized} proved that any algorithm $f$ satisfying four natural axioms has a unique PNE outcome, which is given by the generalized median $\med(y_1,\ldots$ $,y_n,\alpha_0,\ldots,\alpha_n)$, where $\alpha_k= f(0,\ldots,0,\underbrace{1,\ldots,1}_\text{$k$ times})$ for each $k$. 

We note that the `vanilla' $\ell_p$-norm algorithm with no honest agents or regularization satisfies the axioms of \citet{yamamura2013generalized}. Using the result of \citet{yamamura2013generalized} described above, this algorithm has a unique PNE outcome given by the generalized median $\med(y_1,\ldots,y_n,\alpha_0,\ldots,\alpha_n)$, where $ \alpha_k = \frac{k^{\frac{1}{p-1}}}{(n-k)^{\frac{1}{p-1}} + k^{\frac{1}{p-1}}}$ for each $k \in \set{0,1,\ldots,n}$. It is easy to see that $\alpha_0 = 0$ and $\alpha_n = 1$. For $k \in \set{1,\ldots,n-1}$, $\alpha_k$ is the minimizer $\argmin_{\by \in \R} k|1-\by|^{p} + (n-k)|\by|^{p}$. Taking the derivative w.r.t. $\by$, we can see that the optimal solution is given by 
\begin{equation}\label{equation:1d_p_optimal_0_1}
\begin{split}
-k(1-\alpha_k)^{p-1} + (-k)\alpha_k^{p-1} = 0 \implies \alpha_k = \frac{k^{\frac{1}{p-1}}}{(n-k)^{\frac{1}{p-1}} + k^{\frac{1}{p-1}}}
\end{split}
\end{equation}

Below, we extend this to the general $(p,R)$-regression algorithm with $p > 1$, convex regularizer $R$, and with the possibility of honest agents. We omit the proof because, in the next section, we prove this more generally for the linear regression setting (Theorems~\ref{theorem:PNE_existence_regression},~\ref{theorem:unique_outcome}, and~\ref{theorem:NE_is_strategyproof}). 

\begin{theorem}\label{theorem:1d}
	Consider facility location with $n$ agents, of which a subset of agents $M$ are strategic and have single-peaked preferences with peaks at $\vy_{\manip} \in [0,1]^m$. Let $f$ denote the $(p,R)$-regression algorithm with $p > 1$ and convex regularizer $R$. Then, the following statements hold for $f$.
	\begin{enumerate}
		\item For each $\vy_{\manip}$, there is a pure Nash equilibrium $\vhy_{\manip} \in \NE_f(\vy_{\manip})$.
		\item For each $\vy_{\manip}$, all pure Nash equilibria $\vhy_{\manip}  \in \NE_f(\vy_{\manip})$ have the same outcome $f(\vhy_{\manip})$.
		\item There exists a strategyproof algorithm $h$ such that for all $\vy_{\manip}$ and all pure Nash equilibria $\vhy_{\manip} \in \NE_f(\vy_{\manip})$, $f(\vhy_{\manip}) = h(\vy_{\manip})$.
	\end{enumerate}
\end{theorem}

Theorem~\ref{theorem:1d} guarantees the existence of a pure Nash equilibrium and highlights an interesting structure of the equilibrium. The next immediate question is to analyze the quality of this equilibrium. We show that the PPoA of any $(p,0)$-regression algorithm (i.e. without regularization) is $\Theta(n)$. Interestingly, this holds even if only a single agent is strategic, and the bound is independent of $p$.
\begin{theorem}\label{theorem:1d-poa}
	Consider facility location with $n$ agents, of which a subset of agents $M$ are strategic. Let $f$ denote the $(p,0)$-regression algorithm with $p > 1$. When $|\manip| \ge 1$, $\PPoA(f) = \Theta(n)$. 
\end{theorem}

\begin{proof}
	Define $a = \min_{i}{y_i}$ and $b = \max_{i}{y_i}$. As PPoA is measured with MSE, the optimal social cost is achieved with the location $\byh = (1/n) \sum_i y_i$. Let $\byne$ denote the unique PNE outcome of the algorithm. Note that $\byh,\byne$ $\in [a,b]$. For $\byh$, this holds by definition. To see this for $\byne$, WLOG let $\byne < a$.  
	Then all manipulating agents must be reporting $1$, and the honest agents maintain their honest reports in $[a,b]$ (see Lemma ~\ref{lemma:PNE_conditions}). However, then $\ell_p$ loss optimal outcome on this input cannot be $\byne < a$ as $a$ would have a lower loss. A symmetric argument holds for $\byne > b$. Thus, $\byne \in [a,b]$.  
	
	We first show a lower bound of $\Omega(n)$. Suppose a strategic agent $j \in \manip$ has preference with peak at $\alpha_{n-1} = \frac{(n-1)^{\frac{1}{p-1}}}{1 + (n-1)^{\frac{1}{p-1}}}$ and the remaining agents have preferences with peak at $1$. Note that $a=y_j=\alpha_{n-1}$ and $b=1$. We note that a PNE equilibrium is given by $\tilde{y}_j = 0$ and $\tilde{y}_i = 1 \, \forall i \ne j$, regardless of which agents other than $j$ are strategic. By Equation~\eqref{equation:1d_p_optimal_0_1}, the outcome on this input is $a=\alpha_k$. Now, we have that the MSE in the equilibrium is 
	$MSE_{eq} = \sum_i {|y_i - \byne|^2} = (n-1)(b-a)^2,$
	whereas the optimal MSE under honest reports is
	
	\begin{align*}
	MSE_{h} &= \sum_i{|y_i - \byh|^2}\\ 
	&=\left(b - \frac{(n-1)b + a}{n}\right)^2(n-1) + \left(\frac{(n-1)b + a}{n} - a\right)^2\\
	&=\left(\frac{b-a}{n}\right)^2(n-1) + \left(\frac{(n-1)(b-a)}{n}\right)^2\\
	&= \frac{(b-a)^2(n-1) + (n-1)^2(b-a)^2}{n^2}\\ 
	&= \frac{n(n-1)(b-a)^2}{n^2} = \frac{(n-1)(b-a)^2}{n}
	\end{align*}
	Hence, we have that $\PPoA \ge \frac{MSE_{eq}}{MSE_{h}} = n$.
	
	For the upper bound, since the MSE is a strictly convex function with a minimum at the sample mean $\byh$, the maximum allowable value of $MSE_{eq}$ is achieved at one of the end-points $a$ or $b$.
	Hence, we have 
	\begin{align*}
	\PPoA &= \frac{\sum_i{|y_i - \byne|^2}}{\sum_i{|y_i - \bar{y}|^2}}
	\le \max\left\{\frac{\sum_i{|y_i - a|^2}}{\sum_i{|y_i - \bar{y}|^2}},\frac{\sum_i{|y_i - b|^2}}{\sum_i{|y_i - \bar{y}|^2}}\right\}.
	\end{align*}
	We show that each quantity inside $\max$ in the last expression is $O(n)$. Let us prove this for the first quantity. The argument is symmetric for the second. Note that for each $i$ and each $y \in \R$, we have,
	\begin{align*}
	|y_i-y|^2 + |a-y|^2 &\ge |y_i-(y_i+a)/2|^2 + |a-(y_i+a)/2|^2
	= \frac{|y_i-a|^2}{2}.
	\end{align*}
	Hence, we have that for each $i$,
	\begin{equation*}
	|y_i-a|^2 \le 2 \cdot |y_i-\bar{y}|^2 + |a-\bar{y}|^2 \le 2 \sum_i |y_i-\bar{y}|^2.
	\end{equation*}  
	Summing this over all $i$, we get $\frac{\sum_i{|y_i - a|^2}}{\sum_i{|y_i - \bar{y}|^2}} \le 2n$, as desired.
\end{proof}

We remark that both Theorems~\ref{theorem:1d} and~\ref{theorem:1d-poa}, due to their generality, are novel results in the facility location setting.

\section{Linear Regression}\label{section:regression}
We now turn to the more general linear regression setting, which is the focus of our work, and highlight interesting similarities and differences to the facility location setting. Recall that for linear regression, the $(p,R)$-regression algorithm finds the optimal $\vbetast$ minimizing the loss function:
$$\mathcal{L}(\vty,\vx,\vbeta) = \sum_{i=1}^m |\tyi-\vbeta^T \vxi|^p + \sum_{i=m+1}^n |y_i-\vbeta^T \vxi|^p + R(\vbeta)$$ 

Let $i \in \manip$ be a strategic agent. Recall that her outcome is denoted by $\byi = (\vbetast)^T \vxi$. Let $\br_i(\vty_{-i}) = \set{\tyi \in [0,1] : f_i(\tyi,\vty_{-i}) \succeq_i f_i(\tilde{y}'_i,\vty_{-i}) \, \forall \, \tilde{y}'_i \in [0,1]}$ denote the set of her best responses as a function of the reports $\vty_{-i}$ of the other agents. Informally, it is the set of reports that agent $i$ can submit to induce her most preferred outcome.  

\subsection{Properties of the Algorithm, Best Responses, and Pure Nash Equilibria}\label{sec:properties}
We begin by establishing intuitive properties of $(p,R)$-regression algorithms. We first derive the following lemmas.

\begin{lemma}\label{lemma:beta_cant_be_same}
	Fix strategic agent $i \in \manip$ and reports $\vty_{-i}$ of the other agents. Let $\tyi^1$ and $\tyi^2$ be two possible reports of agent $i$, and let $\vbeta^1$ and $\vbeta^2$ be the corresponding optimal regression coefficients, respectively. Then, $\tyi^1 \neq \tyi^2$ implies $\vbeta^1 \neq \vbeta^2$. 
\end{lemma}
\begin{proof}
	Suppose for contradiction that $\vbeta^1 = \vbeta^2 = \vbetast$. We note that at the optimal regression coefficients, the gradient of our strictly convex loss function must vanish. Let the loss functions on the two instances be given by $\mathcal{L}^1$ and $\mathcal{L}^2$, respectively. So for $k \in \set{1,2}$,       \begin{equation*}
	\mathcal{L}^k(\vbeta) = |\tyi^k - \vxi^T\vbeta|^p + \sum_{j\neq i}{|\tyj - \vxj^T\vbeta|^p} + R(\vbeta).
	\end{equation*}
	
	Since $\vbetast$ is optimal for $\mathcal{L}^1$, taking the derivative, we have 
	\begin{align*}
	&\nabla R(\vbetast) - \sum_{j \ne i}{p|\tyj - \vxj^T\vbetast|^{p-2}(\tyj - \vxj^T\vbetast)\vxj} \\
	&\qquad =  p|\tyi^1 - \vxi^T\vbetast|^{p-2}(\tyi^1 - \vxi^T\vbetast)\vxi\\
	&\qquad \ne p|\tyi^2 - \vxi^T\vbetast|^{p-2}(\tyi^2 - \vxi^T\vbetast)\vxi,
	\end{align*}
	where the last inequality follows because $\tyi^1 \neq \tyi^2$ and $\vxi$ is not the $\vec{0}$ vector (its last element is a non-zero constant). Hence, the gradient of $\mathcal{L}^2$ at $\vbetast$ is not zero, which is a contradiction. 
\end{proof}

\begin{lemma}\label{lemma:rearrangement_inequality}
	For $a_1 \ge a_2$, $b_1 \ge b_2$, and $p \ge 1$, we have 
	$$ |a_1-b_1|^p + |a_2-b_2|^p \le |a_1-b_2|^p + |a_2-b_1|^p.$$
\end{lemma}
\begin{proof}
	Note that vector $(a_1-b_2,a_2-b_1)$ majorizes the vector $(a_1-b_1,a_2-b_2)$. For $p \ge 1$, $f(x) = |x|^p$ is a convex function. Hence, by the Karamata majorization inequality, the result follows. 
\end{proof}

\begin{lemma}\label{lemma:monotone_regression}
	The outcome $\byi$ of agent $i$ is continuous in $\vty$, and strictly increasing in her own report $\tyi$ for any fixed reports $\vty_{-i}$ of the other agents.
\end{lemma}
\begin{proof}
	For \emph{continuity}, we refer to Corollary 7.43 in \citet{rockafellar2009variational}, which states that function $F(\vty) = \argmin_{\vbeta}\mathcal{L}(\vty,\vbeta)$ is single-valued and continuous on its domain, when function $\mathcal{L}: \R^{m} \times \R^{n} \rightarrow \R \cup \{-\infty, \infty\}$ is proper\footnote{A function is proper if the domain on which it is finite is non-empty.}, strictly convex, lower semi-continuous, and has $\mathcal{L}^{\infty}(\vec{0},\beta) > 0$, $\forall \vbeta \ne 0$.\footnote{$\mathcal{L}^{\infty}(\vec{0},\beta)$ is known as the horizon function of $\mathcal{L}$.} It is easy to check that our loss function given in Equation~\eqref{equation:loss_fcn} satisfies these conditions. Hence, its minimizer $\vbetast$ is continuous in $\vty$. Since $\vby = \vx\vbetast$, it follows that $\vby$ is also continuous in $\vty$.
	
	For \emph{strict monotonicity}, first note that $\byi =  \vxi^T\vbetast$. Now consider two instances of $(p, R)$-linear regression, $u$ and $w$, that differ only in agent $i$'s reported response, denoted $\tyi^u$ and $\tyi^w$, respectively in the two instances. Hence, $\tyi^u \neq \tyi^w$. Let $\vbeta^u$ and $\vbeta^w$ be the corresponding optimal regression parameters. Without loss of generality, assume $\tyi^u > \tyi^w$, and for contradiction, suppose that $\vxi^T\vbeta^w \geq  \vxi^T\vbeta^u$. Using Lemma~\ref{lemma:beta_cant_be_same}, we get that $\vbeta^u \ne \vbeta^w$. Because our strictly convex loss function has a unique minimizer, we have $\mathcal{L}(\vty^u,\vbeta^u) < \mathcal{L}(\vty^u,\vbeta^w)$ and $\mathcal{L}(\vty^w,\vbeta^w) < \mathcal{L}(\vty^w,\vbeta^u)$. Let us define $\mathcal{C}^u = \sum_{j \ne i}{ |\tyj - \vxj^T\vbeta^u|^p}$ $+ R(\vbeta^u)$ and $\mathcal{C}^w = \sum_{j \ne i}{ |\tyj - \vxj^T\vbeta^w|^p} + R(\vbeta^w)$, we get
	\begin{equation}\label{equation:monotonic_optimal_beta_u}
	|\tyi^u - \vxi^T\vbeta^u|^p + \mathcal{C}^u < |\tyi^u - \vxi^T\vbeta^w|^p + \mathcal{C}^w.
	\end{equation}
	\begin{equation}\label{equation:monotonic_optimal_beta_w}
	|\tyi^w - \vxi^T\vbeta^w|^p + \mathcal{C}^w < |\tyi^w - \vxi^T\vbeta^u|^p + \mathcal{C}^u.
	\end{equation}
	
	Adding Equations~\eqref{equation:monotonic_optimal_beta_w} and \eqref{equation:monotonic_optimal_beta_u}, we have:
	\begin{equation}\label{equation:uu_ww}
	|\tyi^u - \vxi^T\vbeta^u|^p + |\tyi^w - \vxi^T\vbeta^w|^p <
	|\tyi^u - \vxi^T\vbeta^w|^p + |\tyi^w - \vxi^T\vbeta^u|^p
	\end{equation}
	
	Note that because we assumed $\tyi^u > \tyi^w$ and $\vxi^T\vbeta^w \geq  \vxi^T\vbeta^u$, using  Lemma~\ref{lemma:rearrangement_inequality}, we get 
	
	\begin{equation*}
	|\tyi^u - \vxi^t\vbeta^w|^p + |\tyi^w - \vxi^t\vbeta^u|^p \leq |\tyi^u - \vxi^t\vbeta^u|^p + |\tyi^w - \vxi^t\vbeta^w|^p,
	\end{equation*}
	
	which contradicts Equation~\ref{equation:uu_ww}.
\end{proof}

The last lemma demonstrates that $(p,R)$-regression cannot be strategyproof. Consider an instance where each strategic agent $i$ has $y_i \notin \set{0,1}$ and these true data points do not all lie on a hyperplane. Then under honest reporting, not all strategic agents can be perfectly happy, and any agent $i$ with $\byi > y_i$ (or $\byi < y_i$) can slightly decrease (or increase) her report to achieve a strictly more preferred outcome. Next, we show that the best response of an agent is always unique and continuous in the reports of the other agents.

\begin{lemma}\label{lemma:best-response}
	For each strategic agent $i$, the following hold about the best response function $\br_i$.
	\begin{enumerate}
		\item\label{part1} The best response is unique, i.e., $|\br_i(\vty_{-i})|=1$ for any reports $\vty_{-i}$ of the other agents.
		\item $\br_i$ is a continuous function of $\vty_{-i}$.
	\end{enumerate}
\end{lemma}
\begin{proof}
	We first show \emph{uniqueness} of the best response. By Lemma \ref{lemma:monotone_regression}, $f_i$ is continious and strictly increasing in $\tyi$. 
	Consider the minimization problem: $\argmin_{\tyi \in [0,1]}$ ${|y_i - f_i(\tyi, \vty_{-i})|}^p$, where $\vty_{-i}$ is constant. So for now, let us consider $f_i$ to be a function of only $\tyi$. Since $\tyi \in [0,1]$, it achieves a minimum at $a = f_i(0)$ and a maximum at $b = f_i(1)$. If $a \le b \le y_i$, then the minimum of the problem is achieved at $\tyi = 1$. Symmetric case holds for $y_i \le a \le b$ where minimum is achieved at $\tyi = 0$. Lastly, if $y_i \in [a,b]$, by intermediate value theorem, $\exists \, \tyi \, \text{ s.t } \, f_i(\tyi) = y_i$, which is then the minimum. In all cases, the minimum is unique since $f_i$ is strictly increasing. We now show that this unique minimum $\tyi^*$ is indeed the unique best response. If $y_i \in [a,b]$ then reporting $\tilde{y}_i^*$ makes agent $i$ perfectly happy as her outcome matches the peak of her preference, which is clearly best response. If $y_i > b$, then $\tyi^* = 1$ and her outcome is $\byi = b$. Under any other report, her outcome would be $\byi \le b$, which cannot be more preferred. A symmetric argument holds for $y_i < a$ case. 
	
	Now we can use the uniqueness of the best response to argue its \emph{continuity}. More specifically, we want to show that $br_i(\vty_{-i}) = \argmin_{\tyi \in [0,1]} g(\tyi, \vty_{-i})$ is continuous, where $g(\tyi, \vty_{-i}) = |y_i - f_i(\tyi, \vty_{-i})|^p$ is jointly continious due to the continuity of $f_i$. We use the sequence definition of continuity. Fix a convergent sequence $\{\vty_{-i}^n\} \to \vty_{-i}$. Since there is always a unique minimum, the sequence $\{br_i(\vty_{-i}^{n})\}$ is well-defined. We want to show $\{br_i(\vty_{-i}^{n})\} \rightarrow br_i(\vty_{-i})$. By the Bolzano-Weirstrass theorem, every bounded sequence in $\R$ has a convergent sub-sequence. Therefore, this has a convergent sub-sequence $\{br_i(\vty_{-i}^{n_k})\}$ that converges to some $\theta$. Let $br_i(\vty_{-i}) = \theta^*$. We want to first show $\theta = \theta^*$. By the continuity of $g$, $\{g(\theta^*, \vty_{-i}^{n_k})\} \to g(\theta^*, \vty_{-i})$. Also by the minimum, for every individual element of the subsequence $n_k$, we have that $g(\theta^*, \vty_{-i}^{n_k})\ \geq g(br_i(\vty_{-i}^{n_k}), \tilde{y}^{n_k})$. Now again by continuity of $g$, both the above sequences converge and we have: $g(\theta^*, \vty_{-i}) \geq g(\theta, \vty_{-i}))$. Since $\theta^*$ is the unique minimizer for $\vty_{-i}$, we have that $\theta = \theta^*$. So, every convergent sub-sequence of $br_i(\vty_{-i}^n)$ converges to $br_i(\vty_{-i})$. Since this is a bounded sequence, we have that if $\{\vty_{-i}^n\} \to \vty_{-i}$, then $\{br_i(\vty_{-i}^n)\} \rightarrow br_i(\vty_{-i})$. Thus, $br_i$ is continuous.
\end{proof}

We remark that part~\ref{part1} of Lemma~\ref{lemma:best-response} is a strong result: it establishes a unique best response for every possible single-peaked preferences that an agent may have (in fact, our proof shows that this best response depends only on the peak and not on the full preferences). This allows us to avoid further assumptions on the structure of the agent preferences. 

Finally, we derive a simple characterization of pure Nash equilibria in our setting. We show that under a PNE, each strategic agent $i$ must be in one of three states: either she is perfectly happy ($\byi = y_i$), or wants to decrease her outcome ($\byi > y_i$) but is already reporting $\tyi = 0$, or wants to increase her outcome ($\byi < y_i$) but is already reporting $\tyi = 1$.

\begin{lemma}\label{lemma:PNE_conditions}
	$\vty_{\manip}$ is a pure Nash Equilibrium if and only if $(\byi < y_i \land \tyi = 1)\ \lor\ (\byi > y_i \land \tyi = 0)\ \lor \ (\byi = y_i)$ holds for all $i \in \manip$.
\end{lemma}
\begin{proof}
	For the `if' direction, we check that in each case, agent $i \in \manip$ cannot change her report to attain a strictly better outcome. When $\byi < y_i$ and $\tyi = 1$, every other report $\tilde{y}'_i < \tyi = 1$ will result in an outcome $\bar{y'}_i < \byi < y_i$ (Lemma~\ref{lemma:monotone_regression}), which the agent prefers even less. A symmetric argument holds for the $\byi > y_i$ and $\tyi = 0$ case. Finally, when $\byi = y_i$, the agent is already perfectly happy. 
	
	For the `only if' direction, suppose $\vty_{\manip}$ is a PNE. Consider agent $i \in \manip$. The only way the condition is violated is if $\byi < y_i$ and $\tyi \neq 1$ or $\byi > y_i$ and $\tyi \neq 0$. In the former case, Lemma~\ref{lemma:monotone_regression} implies that for a sufficiently small $\epsilon > 0$, agent $i$ increasing her report to $\tilde{y}'_i = 1+\epsilon$ must result in an outcome $\bar{y'}_i \in (\byi,y_i]$, which the agent strictly prefers over $\byi$. This contradicts the assumption that     $\vty_{\manip}$ is a PNE. A symmetric argument holds for the second case.
\end{proof}

Note that Lemma~\ref{lemma:PNE_conditions} immediately implies a na\"{\i}ve but simple algorithm to find a pure Nash equilibrium. Since $\tyi \in \set{0,y_i,1}$ for each $i$, this induces $3^m$ possible $\vty_{\manip}$ vectors. For each such vector, we can compute the outcome of the mechanism $\vby$, and check whether the conditions of Lemma~\ref{lemma:PNE_conditions} are satisfied. This might lead one to believe that the strategic game that we study is equivalent to the finite game induced by the $3^m$ possible strategy profiles. However, this is not true because limiting the strategy set of the agents can give rise to new equilibria which are not equilibria of the original game. We give an explicit example illustrating this below. We further discuss the issue of computing a PNE in Section~\ref{section:implementation_and_experiments}.

\paragraph{Example 1: Finite game leading to different equilibria.}  We use an example from 1D facility location with the average rule --- recall that this is a special case of linear regression --- to illustrate this point.
Consider an example with two agents $1$ and $2$ with true points $y_1 = 0.4$ and $y_2 = 0.5$, respectively, whose preferences are such that each agent $i$ strictly prefers outcome $\by^1$ to $\by^2$ when $|\by^1-y_i| < |\by^2-y_i|$. 

If the agents are allowed to report values in the range $[0,1]$, then the unique PNE of the game is agent $1$ reporting $\ty_1 = 0$ and agent $2$ reporting $\ty_2=1$, and the PNE outcome is $\by=0.5$. 

Now, consider the version with finite strategy spaces, where each agent $i$ must report $\ty_i \in \set{0,1,y_i}$. Suppose the agents report honestly, i.e., $\vty = \vy = (0.4,0.5)$. Then, the outcome is $\by = 0.45$. The only way agent $1$ could possibly improve is by reporting $0$, but in that case the outcome would be $\by=0.25$, increasing $|\by-y_1|$. A similar argument holds for agent $2$. Hence, honest reporting is a PNE of the finite game, but not of the original game.

 \subsection{Analysis of Pure Nash Equilibria}\label{section:regression_analysis}
We are now ready to prove the main results of our work. We begin by showing that a PNE always exists, generalizing the first statement of Theorem~\ref{theorem:1d} from 1D facility allocation to linear regression.

\begin{theorem}\label{theorem:PNE_existence_regression}
	For $p > 1$ and convex regularizer $R$, the $(p, R)$-regression algorithm admits a pure Nash Equilibrium.
\end{theorem}
\begin{proof}
	Consider the mapping $T$ from the reports of strategic agents to their best responses, i.e.,  $T(\tilde{y}_1,\ldots,\tilde{y}_m) = (\br_1(\vty_{-1}),\ldots,$ $\br_m(\vty_{-m}))$. Recall that best responses are unique due to Lemma~\ref{lemma:best-response}. Also, note that pure Nash equilibria are precisely fixed points of this mapping. 
	
	Brouwer's fixed point theorem states that any continuous function from a convex compact set to itself has a fixed point~\cite{pugh2003real}. Note that $T$ is a function from $[0,1]^m$ to $[0,1]^m$, and $[0,1]^m$ is a convex compact set. Further, $T$ is a continuous function since each $\br_i$ is a continuous function (Lemma~\ref{lemma:best-response}). Hence, by Brouwer's fixed point theorem, $T$ has a fixed point (i.e. pure Nash equilibrium).
\end{proof}

Next, we show that there is a unique pure Nash equilibrium outcome (i.e. all pure Nash equilibria lead to the same hyperplane $\vbetast$), generalizing the second statement in Theorem~\ref{theorem:1d}.

\begin{theorem}\label{theorem:unique_outcome}
	For any $p > 1$ and convex regularizer $R$, the $(p,R)$-regression algorithm has a unique pure Nash equilibrium outcome. 
\end{theorem}
\begin{proof}
	Assume by contradiction that there are two equilibria $\vty^1$ and $\vty^2$, which result in distinct outcomes $\vbeta^1$ and $\vbeta^1$, respectively. By Lemma~\ref{lemma:PNE_conditions}, any agent $i$ with $y_i > \max(\by_i^1,\by_i^2)$ or $y_i < \min(\by_i^1,\by_i^2)$ must have the same report in both cases. Similarly, any agent $i$ with $\by_i^2 < y_i < \by_i^1$ must have $\tyi^1 = 0$ and $\tyi^2 = 1$. A symmetric case holds for agents $i$ with $\by_i^1 < y_i < \by_i^2$. Lastly, any agent $i$ with $y_i = \byi^2 < \byi^1$ must have $\tilde{y}_i^2 \in [0,1]$ and $\tilde{y}_i^1 = 0$. Similar arguments hold for the remaining symmetric cases. In all such instances, we note that agents change their reports weakly in the opposite direction to their respective projections. If only one agent changed, Lemma~\ref{lemma:monotone_regression} shows that it leads to a contradiction. We rely on a similar technique to show that multiple agents changing also leads to a contradiction. Note that the only exception to this are agents $k \in \mathcal{D}$, whose preference lies on both hyperplanes (i.e. on their intersection).
	
	Let $\mathcal{A}$ be the set of points who change their reports weakly in the opposite direction as their projections, $\mathcal{D}$ as defined above, and $\mathcal{S}$, the remaining agents who either do not change or are honest. Recall $\byi = \bm{x_i}^T\vbeta$. Then $\forall \, k \, \in \, \mathcal{D} \, , \vxk^T\vbeta^1 = \vxk^T\vbeta^2$ and $\forall \, i \, \in \, \mathcal{A}$:
	\begin{equation}\label{equation:unique_relationship}
	\left(\tyi^1 \geq \tyi^2 \Rightarrow \vxi^T\vbeta^2 \geq \vxi^T\vbeta^1\right) \, \wedge \, \left(\tyi^2 \geq \tyi^1 \Rightarrow \vxi^T\vbeta^1 \geq \vxi^T\vbeta^2\right).
	\end{equation}
	
	Let $\mathcal{C}^1 = \sum_{j \in \mathcal{S}}{ |\tyj - \vxj^T\vbeta^1|^p} + R(\vbeta^1)$ and $\mathcal{C}^2 = \sum_{j \in \mathcal{S}}{ |\tyj - \vxj^T\vbeta^2|^p}$ $+ R(\vbeta^2)$. Noting that $\vbeta^1$ and $\vbeta^2$ uniquely minimize the loss for instances $1$ and $2$, respectively, and $\vbeta^1 \ne \vbeta^2$, we have:
	\begin{align*}
	\sum_{i \in \mathcal{A}}{|\tyi^1 - \vxi^T\vbeta^1|^p} + 
	\sum_{k \in \mathcal{B}}{|\tyk^1 - \vxk^T\vbeta^1|^p} + \mathcal{C}^1 < \sum_{i \in \mathcal{A}}{|\tyi^1 - \vxi^T\vbeta^2|^p} + 
	\sum_{k \in \mathcal{B}}{|\tyk^1 - \vxk^T\vbeta^2|^p} + \mathcal{C}^2,
	\end{align*}
	and
	\begin{align*}
	\sum_{i \in \mathcal{A}}{|\tyi^2 - \vxi^T\vbeta^2|^p} + 
	\sum_{k \in \mathcal{B}}{|\tyk^2 - \vxk^T\vbeta^2|^p} + \mathcal{C}^2 < \sum_{i \in \mathcal{A}}{|\tyi^2 - \vxi^T\vbeta^1|^p} + 
	\sum_{k \in \mathcal{B}}{|\tyk^2 - \vxk^T\vbeta^1|^p} + \mathcal{C}^1.
	\end{align*}
	
	Adding two equations above, we have
	\begin{equation}\label{equation:uniqueness_optimal_sum}
	\sum_{i \in \mathcal{A}}{\left\{|\tyi^1 - \vxi^T\vbeta^1|^p + |\tyi^2 -     \vxi^T\vbeta^2|^p\right\}} < 
	\sum_{i \in \mathcal{A}}{\left\{|\tyi^1 - \vxi^T\vbeta^2|^p + |\tyi^2 -     \vxi^T\vbeta^1|^p\right\}}.
	\end{equation}
	
	Due to Equation~\eqref{equation:unique_relationship}, when we apply Lemma~\ref{lemma:rearrangement_inequality} to each $i \in \mathcal{A}$:
	\begin{equation*}
	|\tyi^1 - \vxi^T\vbeta^2|^p + |\tyi^2 - \vxi^T\vbeta^1|^p \leq 
	|\tilde{y}_i^1 - \bm{x_i}^T\beta^*_1|^p + |\tyi^2 - \vxi^T\vbeta^2|^p.
	\end{equation*}
	
	Thus adding this up for all $i$, we have:
	\begin{equation*}
	\begin{split}
	\sum_{i \in \mathcal{A}}{\left\{|\tyi^1 - \vxi^T\vbeta^2|^p + |\tyi^2 -     \vxi^T\vbeta^1|^p\right\}} \leq 
	\sum_{i \in \mathcal{A}}{\left\{|\tyi^1 - \vxi^T\vbeta^1|^p + |\tyi^2 - \vxi^T\vbeta^2|^p\right\}},
	\end{split}
	\end{equation*}
	which contradicts Equation~\eqref{equation:uniqueness_optimal_sum}.
\end{proof}

While the result above illustrates that the PNE outcome is unique, the equilibrium strategy may not be. This stems from different sets of reports mapping to the same regression hyperplane. In the simplest case, consider the ordinary least squares (OLS) with no regularization, i.e., the $(2,0)$-regression, where all $n$ agents are strategic. Given $\vx \in \R^{d \times n}$, the OLS produces a linear mapping from the reports $\vty$ to the outcomes $\vby$ given by $\vH \vty = \vby$, where $\vH = \vx(\vx^T\vx)^{-1}\vx^T \in \R^{n \times n}$ is a symmetric idempotent matrix of rank $d$ (known as the hat matrix). When $n > d$, $\vH$ is singular, leading to infinitely many $\vty$ which map to the same $\vby$. Of course, they need to still satisfy the conditions of being a PNE (Lemma~\ref{lemma:PNE_conditions}). For a concrete example, if the $n$ true data points lie on a hyperplane, any of the infinitely many reports $\vty$ under which OLS returns this hyperplane --- making all $n$ agents perfectly happy --- is a PNE. 

Given the linear structure of OLS, one wonders if our results can be extended to all \emph{linear mappings}. We say a game is induced by a linear mapping if a matrix $\vH$ relates the agents' outcomes $\vby$ to their reports $\vty$ by the equation $\vH \vty = \vby$. When $\vH$ is a hat matrix arising from OLS, Theorems~\ref{theorem:PNE_existence_regression} and~\ref{theorem:unique_outcome} show that the induced game admits a PNE with a unique outcome. Interestingly, it is easy to show that the proof of Theorem~\ref{theorem:PNE_existence_regression} (existence of PNE) can be extended to all matrices $\vH$. However, there are matricies for which the corresponding game has multiple PNE outcomes. We give an example below. It is an interesting open question to identify the precise conditions on $\vH$ for the induced game to satisfy Theorem~\ref{theorem:unique_outcome} and thus have a unique PNE outcome.

\paragraph{Example 2: Multiple PNE Outcomes in General Linear Mappings} Consider the following matrix: 
\[
\vH =
\begin{bmatrix}
0.8 & -1 \\
-1.2 & 1
\end{bmatrix}
\]
Suppose the agents' preferred values are given by $\vy = (0,0)$. Then, when they report $\vty = (0,0)$, the outcome is $\vby=(0,0)$. This is clearly a PNE as both agents are perfectly happy. When they report $\vty = (1,1)$, the outcome is $\vby = (-0.2,-0.2)$. While neither agent is perfectly happy as the outcome is lower than their preferred value, neither can increase their outcome because they are already reporting $1$. Hence, this is also a PNE with a different outcome.

 \subsection{Connection to Strategyproofness}\label{section:strategyproofness}
A social choice rule maps true preferences of the agents ($\vy$) to a socially desirable outcome ($\vby$ or $\vbetast$). Strategyproofness is a strong requirement: when $f$ is strategyproof, honest reporting is a \emph{dominant strategy} for each agent (i.e., it is an optimal strategy regardless of the strategies of other agents). We say that rule $f$ is \emph{implementable in dominant strategies} if there exists a rule $g$ such that $f(\vy)$ is a dominant strategy outcome under $g$. Although a seemingly weaker requirement (since for a strategyproof rule $f$, one can set $g=f$), the classic revelation principle argues otherwise: if $f$ can be implemented in dominant strategies, then directly eliciting agents' preferences and implementing $f$ must be strategyproof. 

A weaker requirement is that $f$ be \emph{Nash-implementable}, i.e., there exists $g$ such that the Nash equilibrium outcome under $g$ is $f(\vy)$.\footnote{This is weaker because for a strategyproof rule $f$, $f(\vy)$ is a dominant strategy equilibrium outcome (and thus also a Nash equilibrium outcome) under $f$ itself.} Generally, not every Nash-implementable rule is strategyproof. However, a classic line of work in economics~\cite{roberts1979characterization,dasgupta1979implementation,laffont1982nash} proves that Nash-implementable rules are strategyproof for ``rich'' preference domains. It is easy to check that our domain with single-peaked preferences does not satisfy their ``richness'' condition. For single-peaked preferences, we noted in Section~\ref{section:1d} that \citet{yamamura2013generalized} proved such a result in 1D facility location for a family of algorithms with unique PNE outcomes. We extend this to the more general linear regression setting. At this point, we make two remarks. First, the result we establish is stronger than the revelation principle (albeit in this specific domain) as it ``converts'' Nash-implementability (rather than the stronger dominant-strategy-implementability) into strategyproofness. Second, the result of \citet{yamamura2013generalized} for 1D facility location relied on the analytical form of the PNE outcome, so strategyproofness could be explicitly checked. However, the analytical form of the PNE outcome is unknown in the linear regression setting, requiring an indirect argument to establish strategyproofness.
 
We note that our result actually applies to a even broader setting than linear regression: specifically, it applies to any function $f:[0,1]^m \to \R^m$ which has a unique PNE outcome and satisfies an additional condition (stated in the next theorem). We believe that this could have further implications in the theory about implementability of rules, and may be of independent interest. Lastly, as noted by \citet{chen2018strategyproof}, strategyproof mechanisms for linear regression are scarce. This result introduces a new parametric family of strategyproof mechanisms: for given $(p,R)$, the corresponding strategyproof mechanism outputs the unique PNE outcome of $(p,R)$-regression.  

\begin{theorem}\label{theorem:NE_is_strategyproof}
	Let $M$ be a set of agents with $|M|=m$. Each agent $i$ holds a private $y_i \in [0,1]$. Let $f$ be a function which elicits agent reports $\vty \in [0,1]^m$ and returns an outcome $\vby \in \R^m$. Each agent $i$ has single-peaked preferences over $\byi$ with peak at $y_i$. Suppose the following are satisfied:
	\begin{enumerate}
		\item For each $i \in M$ and each $\vty_{-i} \in [0,1]^{m-1}$, $\byi = f_i(\tyi,\vty_{-i})$ is continuous and strictly increasing in $\tyi$. 
		\item For each $\vy \in [0,1]^m$ and each $T \subseteq M$, $f$ has a unique pure Nash equilibrium outcome when agents in $T$ report honestly and agents in $M\setminus T$ strategize. 
	\end{enumerate}
	For $\vy \in [0,1]^m$, let $h(\vy)$ denote the unique pure Nash equilibrium outcome under $f$ when all agents strategize. Then, $h$ is strategyproof. 
\end{theorem}
\begin{proof}
	Let $\vy$ denote the true peaks of agent preferences. To show that $h$ is strategyproof, we need to show that each agent $i$ weakly prefers reporting her true $y_i$ to any other $y'_i$, regardless of the reports $\vy'_{-i}$ submitted to $h$ by the other agents. Fix $\vy'_{-i}$. Let $h_i$ denote the outcome of $h$ for agent $i$. We want to show that $h_i(y_i,\vy'_{-i}) \succeq_i h_i(y'_i, \vy'_{-i})$ for all $y'_i \in [0,1]$. 
	
	Note that $h(y'_i,\vy'_{-i})$ finds the unique PNE outcome under $f$ in the hypothetical scenario where the agents' preferences have peaks at $\vy'$, as opposed to the real scenario in which the peaks are at $\vy$. Let us define a helper function $g_i:[0,1] \to \R$ such that $g_i(\lambda)$ returns the unique PNE outcome for agent $i$ under $f$, when the report of agent $i$ is fixed to $\lambda$ and the other agents strategize according to their preferences $\vy'_{-i}$ and reach equilibrium (this is well-defined due to condition 2 of the theorem). Note that this is independent of agent $i$'s preferences as we fixed her report to $\lambda$. Let $\vhy_{-i}$ be an equilibrium strategy of the other agents in this case. Then, $(\lambda,\vhy_{-i})$ is a PNE under $f$ for all $m$ agents with preferences $\vy'$ if and only if agent $i$ is happy with reporting $\lambda$. The other agents are already happy given agent $i$'s report. Using condition $1$ of the theorem and an argument similar to  Lemma~\ref{lemma:PNE_conditions}, this is equivalent to 
	\begin{equation}\label{eqn:lambda-condition}
	(g_i(\lambda) > y'_i \wedge \lambda = 0) \vee (g_i(\lambda) < y'_i \wedge \lambda = 1) \vee (g_i(\lambda) = y'_i)
	\end{equation}
	
	By condition $2$ of the theorem, we know that for each $y'_i \in [0,1]$, there exists a unique $\lambda^*(y'_i)$ satisfying Equation~\eqref{eqn:lambda-condition}. Note that $h_i(y'_i,\vy'_{-i}) = g_i(\lambda^*(y'_i))$. Using this, we can derive three key properties of the function $g_i$. Let $a = g_i(0)$ and $b = g_i(1)$. 
	\begin{itemize}
		\item $\bm{a \leq b:}$ Assume for contradiction that $a > b$. Choose $y'_i \in (b,a)$. Note that $\lambda = 0$ implies $g_i(\lambda) = a > y'_i$, which satisfies the first clause of Equation~\eqref{eqn:lambda-condition}, while $\lambda = 1$ implies $g_i(\lambda) = b < y'_i$, which satisfies the second clause of Equation~\eqref{eqn:lambda-condition}. Hence, both $\lambda = 0$ and $\lambda = 1$ satisfy Equation~\eqref{eqn:lambda-condition}, which is a contradiction, since $\lambda^*$ is unique. 
		
		\item $\bm{\forall \lambda \in [0,1], g_i(\lambda) \in [a,b]:}$ Assume for contradiction that there exists $\hat{\lambda} \in [0,1]$ such that $g_i(\hat{\lambda}) \notin [a,b]$. WLOG, assume $g_i(\hat{\lambda}) = k < a$ (hence, $\hat{\lambda} \neq 0$). Choose $y'_i = k$. Note that $\lambda = 0$ implies $g(\lambda) = a > k = y'_i$, which satisfies the first clause of Equation~\eqref{eqn:lambda-condition}. Similarly, for $\lambda = \hat{\lambda}$, we have $g_i(\hat{\lambda}) = k = y'_i$, which satisfies the third clause of Equation~\eqref{eqn:lambda-condition}. Hence, both $\lambda = 0$ and $\lambda = \hat{\lambda} \neq 0$ satisfy Equation~\eqref{eqn:lambda-condition}, which is a contradiction. 
		
		\item $\bm{g_i: [0,1] \rightarrow [a,b]}$ \textbf{is surjective/onto:} Assume for contradiction that there exists $\exists c \in (a,b)$ such that $g(\lambda) \neq c$ for any $\lambda \in [0,1]$. Choose $y'_i = c$. Hence, there is no $\lambda$ satisfying the third clause in Equation~\eqref{eqn:lambda-condition}. We see that for $\lambda = 0$, we have $g_i(\lambda) = a < c$, which violates the first clause. Similarly, for $\lambda = 1$, we have $g_i(\lambda) = b > c$, which violates the second clause. Hence, there is no $\lambda$ satisfying Equation~\eqref{eqn:lambda-condition}, which is again a contradiction. 
	\end{itemize}
	
	We are now ready to show that $h_i(y_i,\vy'_{-i}) = g_i(\lambda^*(y_i)) \succeq_i g_i(\lambda^*(y'_i)) = h_i(y'_i, \vy'_{-i})$ for all $y'_i \in [0,1]$. If $y_i \in [a,b]$, then it is easy to see that $\lambda^*(y_i)$ is the unique value which satisfies $g_i(\lambda^*(y_i)) = y_i$ (this exists because $g_i$ is onto). That is, in the equilibrium where agent $i$ reports her true preference, she is perfectly happy. If $y_i < a$, then it is easy to check that $\lambda^*(y_i) = 0$ satisfies Equation~\eqref{eqn:lambda-condition}, and we have $g_i(\lambda^*(y_i)) = a$. Since $g_i(\lambda^*(y'_i)) \in [a,b]$ for any $y'_i$, she will not strictly prefer this outcome. A symmetric argument holds for the $y_i > b$ case. This establishes strategyproofness of $h$.
\end{proof}

\begin{corollary}\label{corollary:regression-PNE-SP}
	Let $f$ denote the $(p,R)$-regression algorithm with $p > 1$ and convex regularizer $R$. Then, there exists a strategyproof algorithm $h$ such that $\forall \vy \in [0,1]^m$ and $\vhy \in \NE_f(\vy)$, $f(\vhy) = h(\vy)$.
\end{corollary}
\begin{proof}
	We already established that the $(p,R)$-regression algorithm satisfies the conditions of Theorem~\ref{theorem:NE_is_strategyproof}. Specifically, $f_i$ is continuous and strictly increasing in the report of agent $i$ (Lemma~\ref{lemma:monotone_regression}). The second condition follows from Theorems~\ref{theorem:PNE_existence_regression} and~\ref{theorem:unique_outcome}, which hold irrespective of which agents are strategic and which are honest. Hence, the result follows immediately from Theorem~\ref{theorem:NE_is_strategyproof}.
\end{proof}

 \subsection{Pure Price of Anarchy}\label{section:regression_PPOA}
So far, our results in linear regression draw conclusions that are similar to those in the 1D facility location setting. We proved that in both cases, a PNE exists, the PNE outcome is unique, and it coincides with the outcome of a strategyproof algorithm. However, there are fundamental differences between the two settings, which we now highlight. The pure price of anarchy is one such difference. In the 1D case, we illustrated that the PPoA is $\Theta(n)$ when no regularizer is used (Theorem~\ref{theorem:1d-poa}). While high, this is still bounded. In linear regression, we will show that the PPoA is unbounded when no regularizer is used. What if we do use a convex regularizer? In practice, the regularizer is often multiplied by a real number $\lambda$, denoting the weight given to regularization, which is tuned by the algorithm designer. We show that for any convex function $R$, the PPoA remains unbounded if $\lambda R$ is used as the regularizer for a large enough $\lambda$. This does leave open the question whether the PPoA might be bounded for some regularizer with a small weight; we leave this for future work. Informally, the next result shows that strategic behavior can make the overall system unboundedly worse-off. 

\begin{figure}[h]
	\centering
	\includegraphics[scale=0.55]{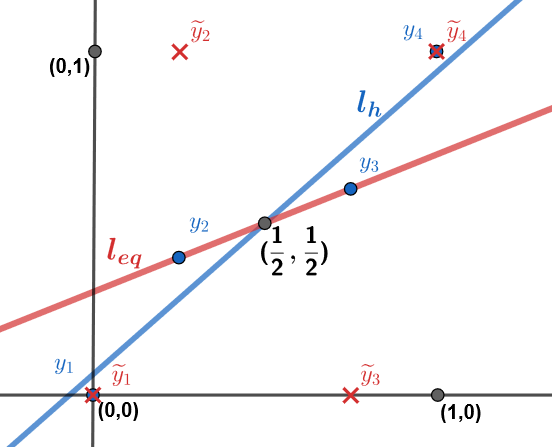}
	\caption{Diagram for Theorem \ref{theorem:regression_sc_unbounded} and Proposition \ref{proposition:infinite_convergence_iterations} with $p=2$ and $R=0$. Blue denotes the honest points and the corresponding line, and red denotes the points at a pure Nash equilibrium and the corresponding equilibrium line.}
	\label{fig:unbounded_diagram}
\end{figure}

\begin{theorem}\label{theorem:regression_sc_unbounded}
	For any $p>1$ and convex regularizer $R$, there exists $\lambda^* > 0$ such that the PPoA of the $(p,\lambda R)$ regression algorithm is unbounded for every $\lambda \ge \lambda^*$. In particular, when there is no regularizer (i.e. $R = 0$), the PPoA of $(p,0)$ regreesion algorithm is unbounded for every $p > 1$.
\end{theorem}
\begin{proof}
	We consider cases depending on whether the regularizer $R$ is constant or not. Starting with the latter, when $R$ is not a constant function, there exist $\vbeta_1$ and $\vbeta_2$ such that $R(\vbeta_1) < R(\vbeta_2)$. Recall that the $(p,\lambda R)$-regression objective is to minimize $\sum_{i=1}^{n}{|\ty_i - \vbeta^Tx_i|^p} + \lambda R(\vbeta)$ given the agent reports $\vty$. Choose $\lambda^* > n$. Note that 
	\begin{align*}
	&\sup_{\vty^1,\vty^2} \left|\sum_{i=1}^{n}{|\ty^1_i - \vbeta^Tx_i|^p}-\sum_{i=1}^{n}{|\ty^2_i - \vbeta^Tx_i|^p}\right| 
	\le \sup_{\vty^1,\vty^2} \left|\sum_{i=1}^{n}{|\ty^1_i - \ty^2_i}|\right| \le n < \lambda^*.
	\end{align*}
	
	We show that the PPoA of $(p,\lambda R)$-regression is unbounded for all $\lambda \ge \lambda^*$. Consider an instance with $n > d$ agents whose honest points all lie on the hyperplane $\vbeta_2$. Let $\vhy$ denote agent reports under some PNE. By our choice of $\lambda^*$, it follows that $\sum_{i=1}^{n}{|\hat{y}_i - \vbeta_1^Tx_i|^p} + \lambda R(\vbeta_1) < \sum_{i=1}^{n}{|\hat{y}_i - \vbeta_2^Tx_i|^p} + \lambda R(\vbeta_2)$ regardless of the value of $\vhy$. Hence, the uniquely optimal hyperplane returned by $(p,\lambda R)$-regression is not $\vbeta_2$, and therefore has non-zero MSE. In contrast, the OLS trivially returns $\vbeta_2$ and has zero MSE, resulting in unbounded PPoA for the $(p,\lambda R)$-regression. 
	
	We now consider the case where $R$ is a constant function. Hence, it does not affect the minimization objective of $(p,\lambda R)$-regression. Thus, without loss of generality, let $R = 0$. We will be using $\byi^p$ to denote the projection of the $(p,0)$-regression equilibrium plane at some $x_i$ and $\vby^p$ for the vector of all projections. We use $\bar{y}^{OLS}_i$ to denote the projection at $x_i$ of the $(2,0)$-regression line using the honest points and $\vby^{OLS}$ for the vector of all such projections. Thus, $\PPoA \ge MSE_{eq}/MSE_h$, where $MSE_{eq} = \sum_i{(y_i - \bar{y}^p_i)^2}$ and $MSE_h = \sum_i{(y_i - \bar{y}^{OLS}_i)^2}$. 
	
	Consider the following example. There are four agents with reported values $(0,0)$, $(\frac{1-\epsilon}{2}, 1)$, $(\frac{1+\epsilon}{2}, 0)$, $(1,1)$. That is, $\vty = (0,$ $\frac{1-\epsilon}{2},$ $\frac{1+\epsilon}{2},$ $1)$. Let the $(p,0)$-regression line for these points pass through $(0,\bar{y}^p_1)$, $(\frac{1-\epsilon}{2}, \bar{y}^p_2), (\frac{1+\epsilon}{2}, \bar{y}^p_3), (1, \bar{y}^p_4)$. By the symmetry of the problem this line must also pass through $(\frac{1}{2}, \frac{1}{2})$. For $p=1$, we have that $\vby^1 = [0, \frac{1-\epsilon}{2}, \frac{1+\epsilon}{2}, 1]$. Note that the residuals for points 2 and 3 are higher than for points 1 and 4, and observe that for $p>1$, the $(p,0)$-linear regression algorithm progressively tries to minimize the larger residuals. One can check that for $p>1$, $\bar{y}^p_2 = \bar{y}^1_2 + a = \frac{1-\epsilon}{2} + a$ and $\bar{y}^p_3 = \bar{y}^1_3 - a = \frac{1+\epsilon}{2} - a$ for some $a>0$. Since all $\ell_p$-regression lines pass through $(\frac{1}{2}, \frac{1}{2})$, by similar triangles we have that for $p>1$, $\bar{y}^p_1 = \bar{y}^1_1 + \frac{a}{\epsilon} = \frac{a}{\epsilon}$. Now if the preferred/true values of the 4 agents are $\vy = (0,\bar{y}^{p}_2,\bar{y}^{p}_3,1)$, the reported values above are a pure Nash Equilibrium, and the projection values are unique (by Theorem \ref{theorem:unique_outcome}). Note this is regardless of whether agents 1 and 4 are strategic or honest. As such, we have $MSE_{eq} = 2\left(\frac{a}{\epsilon}\right)^2$. 
	
	For $MSE_h$, note that the hat matrix for $(2,0)$-regression depends only on $\vx$, and has the form $\vH = \vx(\vx^T\vx)^{-1}\vx^T$ and $\vby^{OLS} = \bm{H}\vy$. The symmetry of the honest points for any $p$ means that the $(2,0)$-regression line always passes through $(\frac{1}{2}, \frac{1}{2})$ as well. For $p=1$, the honest points are co-linear, meaning the $(2,0)$-regression line of these points have 0 residual for all points (in fact, it's the same as the equilibrium line). For $p>1$, as we mentioned above, honest points 2 and 3 adjust by some $a$ and we have $\vy = (0,\bar{y}^{1}_2 + a, \bar{y}^{1}_3 - a, 1)$. We now consider the affect of these two changed honest points on the residual at $x_1$ and $x_2$. That is, we consider $r^h_1 = |\bar{y}^{OLS}_1 - y_1|$ and $r^h_2 = |\bar{y}^{OLS}_2 - y_2|$ respectively - noting a symmetric case exists for $r^h_3$ and $r^h_4$. First, we have the following values for the matrix $H$:
	
	\begin{equation}
	\begin{split}
	\bm{H}_{12} = \bm{H}_{21} = \frac{(1+\epsilon)^2}{4(1+\epsilon^2)} \quad \bm{H}_{13} = \bm{H}_{31} = \frac{(\epsilon-1)^2}{4(1+\epsilon^2)} \\
	\bm{H}_{22} = \frac{3\epsilon^2 + 1}{4(1+\epsilon^2)} \quad \bm{H}_{23} = \bm{H}_{32} = \frac{1-\epsilon^2}{4(1+\epsilon^2)}
	\end{split}
	\end{equation}
	
	Note that $r^h_1 = r^h_2 = 0$ when $p=1$, and only $y_2$ and $y_3$ have changed (by $+a$ and $-a$ respectively) for $p \ne 1$. Recall, $y_1 = 0$ and $y_2 = \bar{y}^p_2 = \bar{y}^1_2 + a$. Denote the $i^{th}$ row of $H$ by $\bm{h_i}$. Then we have:
	
	\begin{equation*}
		r^h_1 = \bm{h_1} \cdot
		\begin{bmatrix}
		0      \\
		\bar{y}^{1}_2 + a      \\
		\bar{y}^{1}_3 - a      \\
		1
		\end{bmatrix} - 0 
		= \bm{h_1} \cdot
		\begin{bmatrix}
		0      \\
		\bar{y}^{1}_2      \\
		\bar{y}^{1}_3      \\
		1
		\end{bmatrix} 
		+ \bm{h_1} \cdot
		\begin{bmatrix}
		0      \\
		a      \\
		-a      \\
		0
		\end{bmatrix}
		= \bm{h_1} \cdot
		\begin{bmatrix}
		0      \\
		a      \\
		-a      \\
		0
		\end{bmatrix}
	\end{equation*}
	
	$$\therefore r^h_1 = a\frac{(1+\epsilon)^2}{4(1+\epsilon^2)} - a\frac{(\epsilon-1)^2}{4(1+\epsilon^2)} = \frac{a\epsilon}{(1+\epsilon^2)}$$
	Similarly, for $r^h_2$, we have that: 
	\begin{equation*}
		\begin{split}
		r^h_2 = \left(\frac{1-\epsilon}{2} + a\right) - \bm{h_2} \cdot
		\begin{bmatrix}
		0      \\
		\bar{y}^{1}_2 + a      \\
		\bar{y}^{1}_3 - a      \\
		1
		\end{bmatrix}
		= \left(\frac{1-\epsilon}{2} + a\right) - \left(
		\frac{1-\epsilon}{2}
		+ \bm{h_2} \cdot
		\begin{bmatrix}
		0      \\
		a      \\
		-a      \\
		0
		\end{bmatrix} \right)
		\end{split}
	\end{equation*}
	
	$$\therefore r_2 = a - a\frac{3\epsilon^2 + 1}{4(1+\epsilon^2)} + a\frac{1-\epsilon^2}{4(1+\epsilon^2)} = \frac{a}{1+\epsilon^2}$$
	
	By symmetry, $r_1 = r_4$ and $r_2 = r_3$. 
	Thus, we have that the PPoA of $(p,0)$-regression satisfies: 
	$$ \PPoA \ge \frac{2\left(\frac{a}{\epsilon}\right)^2}{2\left[\left(\frac{a\epsilon}{(1+\epsilon^2)}\right)^2 + \left(\frac{a}{1+\epsilon^2}\right)^2\right]} = \frac{\frac{1}{\epsilon^2}}{\frac{1}{1 + \epsilon^2}} = 1 + \frac{1}{\epsilon^2}$$
	
	As $\epsilon \rightarrow 0$, the PPoA becomes unbounded.
\end{proof}

\section{Implementation and Experiments}\label{section:implementation_and_experiments}
While the main goal of this paper is to understand the structure of pure Nash equilibria under linear regression, one might wonder whether, given honest inputs, the unique PNE outcome can be computed efficiently. In this section, we briefly examine this, discover another aspect in which linear regression departs from 1D facility location, and describe some interesting phenomena regarding the PPoA of $(p,R)$-regression mechanisms in practice. We leave detailed computational and empirical analysis of $(p,R)$-regression to future work.

\subsection{Computation of Pure Nash Equilibria}
In facility location, a fully constructive characterization of strategyproof algorithms is known~\cite{moulin1980strategy}. This, along with Theorem~\ref{theorem:1d} and a formula of \citet{yamamura2013generalized}, allows easy computation of the PNE outcome of any $(p,R)$-regression; details are in section \ref{section:1d}. However, characterizing strategyproof algorithms is a challenging open question for the linear regression setting~\cite{chen2018strategyproof}. Thus, while Theorem~\ref{theorem:NE_is_strategyproof} demonstrates that the PNE outcome is also the outcome of a strategyproof algorithm, it does not allow us to derive an analytic expression for the unique PNE outcome. 

In Section~\ref{sec:properties}, we outlined an exponential-time approach that follows immediately from Lemma~\ref{lemma:PNE_conditions}. However, this is impractical unless there are very few agents. Turning elsewhere, a standard approach to computing Nash equilibria is through best-response updates~\cite{ben2017best,ben2019regression,yamamura2013generalized}. Specifically, we start from an (arbitrary) profile of reports by the agents, and in each step, allow an agent not already playing her best response, to switch to her best response. If this process terminates, it must do so at a PNE, regardless of initial conditions. For 1D facility location, it is easy to show that this terminates at a PNE in finitely many steps (see below). For linear regression, however, we show in Proposition \ref{proposition:infinite_convergence_iterations} that the process need not terminate in finitely many steps even for the most simple OLS algorithm.

\paragraph{Best response dynamics converges in finite iterations for 1d facility location} We give an informal argument that under the average rule in 1D, starting from any reports, there is always a best response path that terminates at a PNE in finitely many iterations. For $n$ agents (of which $m$ are strategic), to move the mean by an amount $\Delta$, an agent has to move their report by an amount $n\Delta$. Now fix an initial set of reports. Consider only the 2 strategic agents with the lowest and the highest preferred values, say these are $y_1$ and $y_m$, respectively. Consider best response updates by only one of these two agents. If initially $\by \notin [y_1, y_m]$, both agents increase their reports until $\by \in [y_1,y_m]$. The only case where this does not happen is if both agents become saturated by reporting $1$. If they do bring $\by \in [y_1,y_m]$, then after each move of agent $1$: (a) she is perfectly happy, causing the agent $m$ to move up by $n(y_m-y_1)$ or become saturated at $\tilde{y}_m = 1$, or (b) she goes to $0$ and becomes saturated. Hence, in each iteration, either one agent moves (in a constant direction) by at least $n(y_m-y_1)$, or one agent becomes saturated. Hence, in finitely many steps, either agent $1$ is saturated at $0$ with $\by \ge y_1$ or agent $m$ is saturated at $1$ with $\by \le y_m$. It is easy to see that this agent will never move again. We can now ignore the saturated agent, and repeat the process with the remaining $m-1$ strategic agents. Using this approach inductively, it follows that an equilibrium will be reached in finitely many iterations.

\begin{proposition}\label{proposition:infinite_convergence_iterations}
	For the OLS (i.e. $(2,0)$-regression algorithm), there exists a family of instances in which no best-response path starting from honest reporting terminates in finite steps. 
\end{proposition}
\begin{proof}
	Consider the $4$ agent setting (also used in Theorem~\ref{theorem:regression_sc_unbounded}) illustrated in Figure~\ref{fig:unbounded_diagram}. That is, let the preferred/true values be: $(0,0)$, $(\frac{1-\epsilon}{2}, y_2)$, $(\frac{1+\epsilon}{2}, y_3)$, $(1,1)$, where $y_2$ and $y_3$ are such that when $\vty = [0, 1, 0, 1]$, the corresponding projections are: $\bar{y}_2 = y_2$ and $\bar{y}_3 = y_3$. Thus, $\vty = [0,1,0,1]$ is an equilibrium strategy. Let agents 2 and 3 be strategic.\footnote{Whether agents $1$ and $4$ are strategic or honest does not matter in this example.} Since $p=2$, we have a linear mapping characterized by $\vH\vty = \vby$. $H_{ij}$ reflects the effect $\tyi$ has on $\bar{y}_j$, and $\vH$ is symmetric. By strong monotonicity (Lemma~\ref{lemma:monotone_regression}), $H_{ii}$ is always positive. It is easy to compute that $H_{23} = H_{32}= \frac{1 - \epsilon^2}{4(1 + \epsilon^2)} > 0$. Let the agents initially start by reporting honestly, and as such $\bar{y}_2 < y_2 = \tilde{y}_2$ and $\bar{y}_3 > y_3 = \tilde{y}_3$.
	
	Since there are only 2 strategic agents, they take turns playing best response alternatively. Consider a round in which agent $2$ plays best response, and at the start of the round, the following hold: (1) $\tilde{y}_2 \geq y_2$, $\tilde{y}_3 \leq y_3$, and (2) $\bar{y}_2 \leq y_2$ and $\bar{y}_3 \geq y_3$. Since agent 2 is playing best response, she is not perfectly happy. Hence, $\bar{y}_2 < y_2$. Thus, by Lemma~\ref{lemma:monotone_regression}, agent 2 must increase $\tilde{y}_2$ by some $a > 0$. Since $H_{23} > 0$, this maintains $\bar{y}_3 > y_3$. Similarly, when agent 3 plays a best response, it maintains $\bar{y}_2 < y_2$. Since the initial conditions (honest reporting) satisfy (1) and (2), they will always be satisfied. That is, player $2$ will always report less than $1$ and have $\bar{y}_2 < y_2$, and player $3$ will always report greater than $0$ and have $\bar{y}_3 > y_3$. Thus, the PNE will never be reached in finitely many steps. 
	
	To see this formally, consider a stage satisfying (1) and (2) wherein the best response of agent 2 is $\tilde{y}_2 = 1$ and $\tilde{y}_3 \ne 0$. Since this is a best response, $\bar{y}_2 \leq y_2$ (< in case she isn't perfectly happy) and thus $\bar{y_3} > y_3$. If agent 3 now under-reports and plays $\tilde{y}_3 = 0$, then since $H_{32} > 0$, $\bar{y}_2 < y_2$. However, we now have $\vty = (0,\, 1,\, 0,\, 1)$ where we know the outcome is: $\bar{y}_2 = y_2$ and $\bar{y}_3 = y_3$. Since the regression outcome is unique, this is a contradiction. A similar situation hold for agent 3. Thus if $\tilde{y_3} \ne 0$, best response of agent 2, $\ne 1$ and if $\tilde{y}_2 \ne 1$, the best response of agent 3, $\ne 0$. Since these conditions hold initially, they hold in all rounds.
	
	Thus starting from honest values, agent 2 always over-reports and 3 under-reports and the outcome is never the unique equilibrium outcome. Moreover, at no round does agent 2 or 3 ever reach their equilibrium strategy. Thus at this initial value, no possible best response sequence will terminate in finite iterations.
\end{proof}

However, we emphasize that the example in the proof of Proposition~\ref{proposition:infinite_convergence_iterations} is a worst-case example. In practice, best-response update works quite well for finding the unique PNE outcome quickly; we use this approach successfully in the experiments described next. 

 \subsection{Experiments}
We conduct experiments on both synthetic data and real data to measure two aspects of strategic manipulation: the number of best-response updates needed to reach a pure Nash Equilibrium (red line) and the average PPoA\footnote{We abuse the terminology slightly for simplicity. The average PPoA refers to the average ratio of the loss under the PNE outcome of a mechanism to the loss under the OLS with honest reporting in our experiments.} (solid blue line), which we compare against the average PPoA of the strategyproof LAD (i.e. $(1,0)$-regression) algorithm. We focus on four key parameters: the number of agents $n$, the dimension of explanatory variables $d$, the norm value $p$, and the fraction of agents who are strategic, denoted $\alpha = m/n \in [0,1]$. The regularizer $R$ is always set to 0. We also vary the norm $q$ (default is $q=2$) with regards to which the loss is measured in the PPoA definition. To find the unique PNE outcome, we used best-response updates to obtain outcome they converged to, and verified that it was a PNE (and it always was). 
\begin{figure*}[htb]
	\centering
	\captionsetup[subfigure]{justification=centering}
	\begin{subfigure}{.32\textwidth}
		\centering
		\includegraphics[scale=0.25]{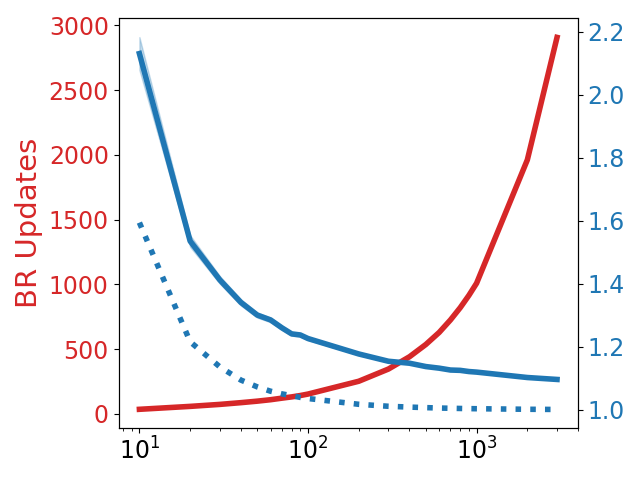}
		\caption{$n \in [10,3000]$}
		\label{figure:1a}
	\end{subfigure}
	\begin{subfigure}{.32\textwidth}
		\centering
		\includegraphics[scale=0.25]{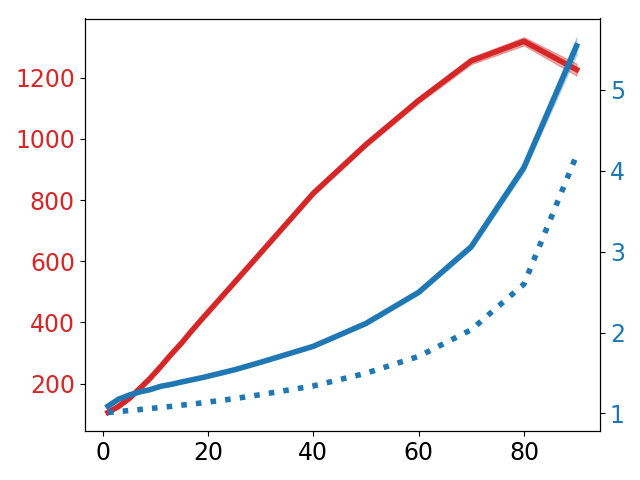}
		\caption{$d \in [1,90]$}
		\label{figure:1b}
	\end{subfigure}
	\begin{subfigure}{.32\textwidth}
		\centering
		\includegraphics[scale=0.25]{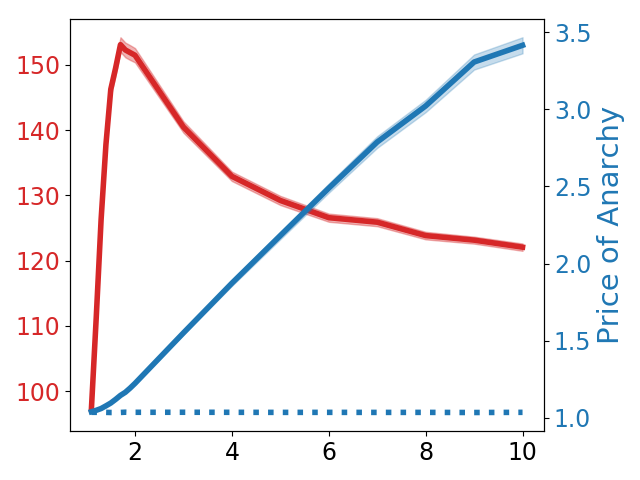}
		\caption{$p \in [1.1,10]$}
		\label{figure:1c}
	\end{subfigure}
	\caption{The effect of varying $n$, $d$, and $p$ on synthetic data with 95\% confidence intervals}
\end{figure*}
\begin{figure*}[htb]
	\centering
	\captionsetup[subfigure]{justification=centering}
	\begin{subfigure}{.32\textwidth}
		\centering
		\includegraphics[scale=0.25]{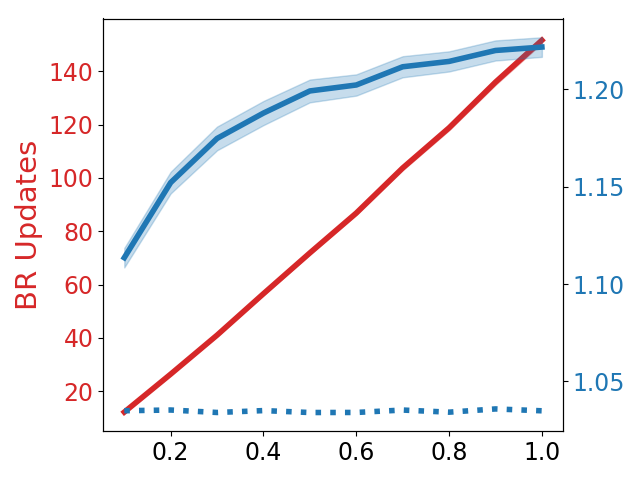}
		\caption{$\alpha\in[0.1,1]$, synthetic data}
		\label{figure:2a}
	\end{subfigure}
	\begin{subfigure}{.32\textwidth}
		\centering
		\includegraphics[scale=0.25]{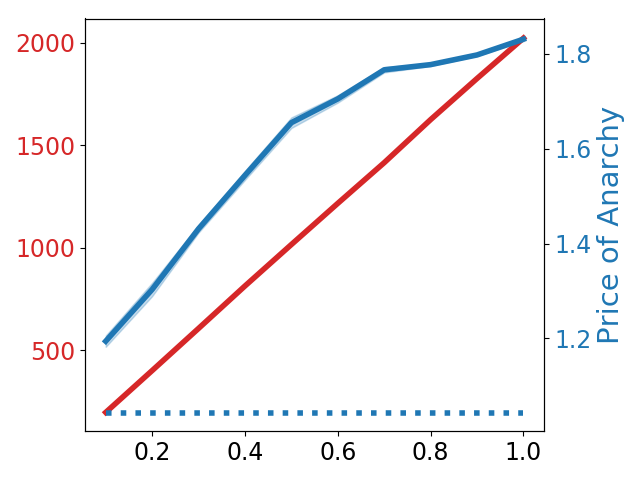}
		\caption{$\alpha\in[0.1,1]$, Kaggle dataset}
		\label{figure:2b}
	\end{subfigure}
	\begin{subfigure}{.32\textwidth}
		\centering
		\includegraphics[scale=0.25]{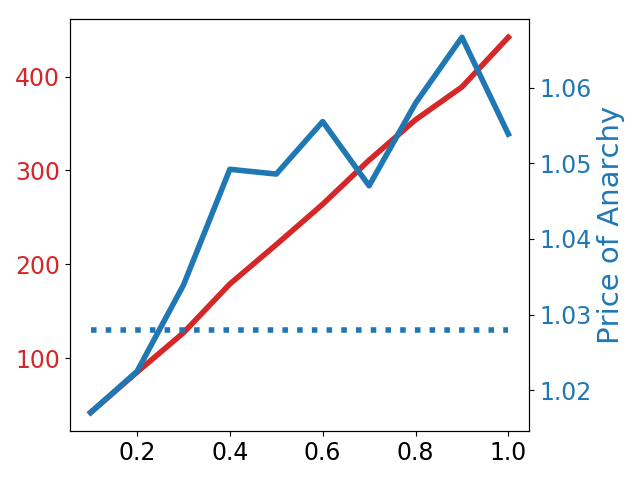}
		\caption{$\alpha\in[0.1,1]$, UCI dataset}
		\label{figure:2c}
	\end{subfigure}
	\caption{The effect of varying $\alpha$ on synthetic and real data. Plots with synthetic data have 95\% confidence intervals.}
\end{figure*}

\paragraph{Synthetic experiments:} In each experiment, we vary one parameter, while using default values for the others. The default values are $n=100$, $d=6$, $p=2$, and $\alpha = 1$.\footnote{We choose $\alpha=1$, which corresponds to all agents being strategic, as the default value in our experiments because this is the standard setting studied in the game-theoretic literature on regression. Note that our theoretical results allow some of the agents to be honest.} We plot the average results over $1,000$ random instances along with $95\%$ confidence bounds (although they are too narrow to be visible in most plots). The data generation process is as follows. First, we sample $\vbetast \in [-1,1]^{d+1}$ uniformly at random. Next, we sample each entry in $\vx \in \R^{d \times n}$ iid from the standard normal distribution and set each $y_i = (\vbetast)^T x_i + \epsilon_i$, where $\epsilon_i$ is Gaussian noise with zero mean and s.d. $0.5$. Finally, we normalize $\vy$ to lie in $[0,1]^n$. 

\paragraph{Real experiments:} We also conduct experiments with two real-world housing datasets: the California Housing Prices dataset from Kaggle with $n \approx 2000$ and $d=9$ (Figure \ref{figure:2b}) and the real estate valuation dataset from UCI with $n \approx 400$ and $d=7$ (Figure \ref{figure:2c})~\cite{kaggle_california,yeh2018building}. In these experiments, we also normalize $\vy$ to lie in $[0,1]^n$.

Figures \ref{figure:1a}, \ref{figure:1b}, \ref{figure:1c} and \ref{figure:2a} show the effect of varying $n$, $d$, $p$, and $\alpha$, respectively, in our synthetic experiments. With a higher number of agents $n$, the best-response process takes longer, but the PPoA decreases quickly. The dependence on $d$ is more interesting. For $d < n$, the number of best-response steps and the PPoA increase with $d$ (with a slight decrease in the former and a quicker increase in the latter as $d$ approaches $n=100$). Of course, when $d=n$, the only PNE is where all agents are perfectly happy, which means the number of best-response steps drop to zero and PPoA drop to $1$. Hence, for $d < n$, there is a curse of dimensionality, even though $d=n$ is an ideal scenario. 

The effect of $p$ is also surprising. With $p \in (1,2]$, intuitively, one would expect a tradeoff. Mechanisms with $p$ closer to $1$ may be less vulnerable to manipulation than the OLS ($p=2$); indeed, $p=1$ is known to be strategyproof. But given the equilibrium reports, OLS at least minimizes the MSE, which is the objective underlying our PPoA definition, whereas mechanisms with $p < 2$ optimize a different objective. Given this, we find it surprising that, not only does $p < 2$ result in a lower PPoA than $p=2$, but PPoA seems to increase monotonically with $p$ (Figure~\ref{fig:ppoa_plot} below shows that this is also true when PPoA is measured using the $q$-norm for other values of $q$). We also note that the strategyproof $(1,0)$-regression algorithm performs no worse than the PNE of the $(p,0)$-regression algorithm for any $p > 1$ in terms of MSE. Another observation of note is that the number of best-response updates increases until $p \approx 2$ and then decreases. In our synthetic and real experiments, both the number of best-response updates and the PPoA generally increase with $\alpha$, which is expected. However, it is worth noting that in Fig. 2b, even as few as $10\%$ of the agents strategizing leads to a $27\%$ increase in the overall MSE, and with all agents strategizing, the MSE doubles. In Fig 2c, the effect of strategizing is more restrained. Surprisingly, in this case, the OLS equilibrium outperforms the $(1,0)$-regression algorithm for small $\alpha$.

\paragraph{Experiment - PPoA with different $q$}
So far we consider PPoA measured with respect to mean squared error ($q=2$), which is the squared $\ell_2$ norm of residuals. We now experimentally evaluate PPoA measured with respect to other values of $q$, as defined below:
$$
\PPoA_q(f) = \max_{\vy \in [0,1]^n} \frac{\max_{\vby \in \NE_f(\vy)} {\textstyle\sum_{i=1}^n} |y_i-\byi|^q}{{\textstyle\sum_{i=1}^n} |y_i-\byi^{q\text{-opt}}|^q},
$$
where $\byi^{q\text{-opt}}$ is the outcome of the mechanism minimizing $\ell_q$ norm of residuals with honest reports. 

Figure~\ref{fig:ppoa_plot} shows $\PPoA_q$ for different $\ell_p$ regression algorithms. Once again, we notice the same pattern for each value of $q$ as we did in Figure 1c for $q=2$: the PPoA increases monotonically with $p$.

\begin{figure}[htb]
	\centering
	\includegraphics[scale=0.45]{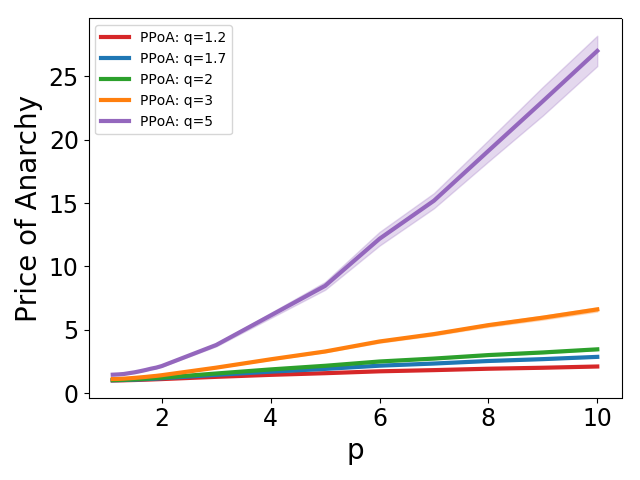}
	\caption{Varying $p$ between $1.1$ and $10$ and graphing the PPoA using different values of $q$. The same defaults are used as in other synthetic experiments ($n=100$, $d=6$, $\alpha=1$) and the average of $1000$ random instances are plotted with $95\%$ confidence intervals (though too narrow to be visible on some curves).}
	\label{fig:ppoa_plot}
\end{figure}

\section{Discussion and Future Work}
This work focused on the role of \textit{strategic noise} in linear regression, where data sources manipulate their inputs to minimize their own loss. We established that a popular class of linear regression algorithms --- minimizing the $\ell_p$ loss with a convex regularizer --- has a unique pure Nash equilibrium outcome. Our theoretical results show that in the worst case, strategic behavior can cause a significant loss of efficiency, but experiments highlight a less pessimistic average case, which future work can focus on rigorously analyzing.

It is also interesting to ponder the implications of our general result connecting strategyproof algorithms to the unique PNE of non-strategyproof algorithms beyond linear regression. Similar results are known in other domains~\cite{roberts1979characterization,dasgupta1979implementation,laffont1982nash}, including unique equilibria of first-price auctions~\cite{chawla2013auctions}. This indicates the possibility of a more general result along these lines.

Lastly, the study of strategic noise in machine learning environments is still in its infancy. We view our work as not only advancing the state-of-the-art, but also as a stepping stone to more realistic analysis. For example, future work can move past assuming that agents have complete information about others' strategies --- a common assumption in the literature~\citep{dekel2010incentive,ben2017best,ben2019regression} --- and consider Bayes-Nash equilibria. Considering other equilibrium concepts relevant to machine learning settings may also prove fruitful. Other extensions include studying non-strategyproof algorithms in environments such as classification or generative modeling, and investigating generalization of equilibria (i.e. whether the equilibrium with many agents can be approximated by sampling a few agents).

\bibliographystyle{plainnat}  
\bibliography{abb,bibliography}  

\end{document}